\documentclass[11pt,a4paper]{article}
\pdfoutput=1
\usepackage{graphicx,amssymb,amsmath}
\usepackage{epsfig}
\usepackage{epsf}
\usepackage{amsthm}
\usepackage{amsfonts,amssymb}
\usepackage{verbatim}
\usepackage{listings}
\usepackage{fixmath}
\usepackage{graphicx}
\usepackage{tikz}

\usetikzlibrary{shapes.geometric}

\newtheorem{theorem}{Theorem}[section]
\newtheorem{lemma}[theorem]{Lemma}

\newtheorem{corollary}[theorem]{Corollary}
\newtheorem{proposition}[theorem]{Proposition}

\theoremstyle{definition}
\newtheorem{definition}[theorem]{Definition}
\theoremstyle{remark}
\newtheorem{remark}[theorem]{Remark}

\renewcommand*{\Phi}{\varPhi}%
\newcommand*{\floor}[1]{\lfloor#1\rfloor}%
\newcommand*{\chicf}{\chi_{\text{\textrm{\textup{cf}}}}}%
\newcommand*{\chium}{\chi_{\text{\textrm{\textup{um}}}}}%
\newcommand*{\isminorof}{\preccurlyeq}%
\newcommand*{\containssubgraph}{\supseteq}%
\newcommand*{\subgraph}{\subseteq}%
\newcommand*{\issubgraphof}{\subgraph}%
\newcommand*{\propersubgraph}{\subset}%
\newcommand*{\card}[1]{\lvert#1\rvert}%
\newcommand*{\Gup}{\hat{G}}%
\newcommand*{\Gdown}{\check{G}}%
\newcommand*{\vup}{\overline{v}}%
\newcommand*{\vdown}{\smash[b]{\underline{v}}}%
\newcommand*{\Pup}{P^{\uparrow}}%
\newcommand*{\Pdown}{P^{\downarrow}}%
\newcommand*{\vcs}{v^{\text{\textrm{\textup{cs}}}}}%
\newcommand*{\vp}{v^{\text{\textrm{\textup{p}}}}}%
\newcommand*{\issubsetof}{\subseteq}%
\newcommand*{\isstrictsubsetof}{\subset}%
\newcommand*{\logfhol}{18\log_{2}}%

\begin{document}

%\begin{titlepage}
\title{Graph unique-maximum and conflict-free colorings}

\author{%
Panagiotis Cheilaris 
\\
Center for Advanced Studies in Mathematics\\
Ben-Gurion University\\
panagiot@math.bgu.ac.il\\
\and G{\'e}za T{\'o}th\thanks{Supported by OTKA T 038397 and 046246.} \\
R{\'e}nyi Institute\\
Hungarian Academy of Sciences\\
geza@renyi.hu 
}

\date{}%
\maketitle

\begin{abstract}
We investigate the relationship 
between two kinds of vertex colorings of graphs:
unique-maximum colorings and 
conflict-free colorings.
In a unique-maximum coloring, the colors are ordered, and 
in every path of the graph the
maximum color appears only once.
%occurring in a vertex of the path is unique in the path.
In a conflict-free coloring, 
in every path of the graph there is 
a color that appears only once.
%occurring in a vertex of the path such that it occurs
%nowhere else in the path.
%These colorings have applications in parallelization of some
%processes
%and frequency assignment in cellular networks.
We also study computational complexity aspects of
conflict-free colorings and prove a completeness result. 
%We introduce two games in graphs that are closely related to the
%conflict-free and unique-maximum chromatic numbers.
Finally, we improve lower bounds for those chromatic numbers 
of the grid graph.
\end{abstract}

\bigskip


\section{Introduction}

In this paper we study two types of vertex 
colorings of graphs, both related to 
paths.
The first one is the following:

\begin{definition} \label{def:ordcol}
A \emph{unique-maximum coloring with respect to paths}
of $G=(V,E)$ with $k$ colors
is a function $C \colon V \to \{1,\dots,k\}$
such that for  each path $p$ in $G$
the maximum color occurs exactly once on the vertices of $p$.
The minimum $k$ for which a graph $G$ has a unique-maximum coloring
with $k$ colors
is called the \emph{unique-maximum chromatic number} of $G$ and is
denoted by $\chium(G)$.
\end{definition}

Unique maximum colorings are known alternatively in the literature as
\emph{ordered colorings} or \emph{vertex rankings}.
The problem of computing unique-maximum colorings is a well-known and
widely studied problem (see e.g.\ \cite{orderedcoloring}) with many
applications including \emph{VLSI design}
\cite{DBLP:conf/focs/Leiserson80} and \emph{parallel Cholesky
decomposition} of matrices \cite{liu:134}. 
The Cholesky decomposition method is used in solving sparse linear
systems $Ax=b$, whenever $A$ is a symmetric $n\times n$ 
positive-definite matrix, and is faster than the more general LU
decomposition.
In \cite{liu:134}, given a symmetric $n\times n$ 
positive-definite matrix $A$,
a graph $G(A)$ on $n$ vertices is defined which encodes the
data dependencies between different columns in the linear system.
The unique-maximum chromatic number of $G(A)$ is a rough estimate
of the work required in parallel Cholesky decomposition of matrix $A$.
The 
problem is also interesting for the Operations Research community,
because it has applications in \emph{planning efficient assembly of
products} in manufacturing systems \cite{optnoderanktree}. In
general, it seems that the vertex ranking problem can model situations
where interrelated tasks have to be accomplished fast in parallel
(assembly from parts, parallel query optimization in databases,
etc.)
Another application of unique-maximum colorings is in 
estimating the worst-case complexity of
\emph{finding local optima in neighborhood structures}. 
A neighborhood structure is a connected graph in which
every vertex has a real value. 
Suppose that we want to find a vertex $v$ which is a local optimum. For
example, if $v$ is a local minimum, then its value is not greater
than the values of its adjacent vertices. The goal is to query as few
vertices of the neighborhood structure as possible.
In some classes of bounded-degree neighborhood structures (like
grids), 
the worst-case complexity of finding a local optimum is 
closely related to the unique-maximum chromatic number of the
corresponding graph
(see \cite{LTT1989dam}).

The other type of vertex coloring can be seen as a relaxation of the
unique-maximum coloring.

\begin{definition} \label{def:cfcol}
A \emph{conflict-free coloring with respect to paths}
of $G=(V,E)$ with $k$ colors
is a function $C \colon V \to \{1,\dots,k\}$
such that for each path $p$ in $G$
there is a color that occurs exactly once on the vertices of $p$.
The minimum $k$ for which a graph $G$ has a conflict-free coloring
with $k$ colors
is called the \emph{conflict-free chromatic number} of $G$ and is
denoted by $\chicf(G)$.
\end{definition}

Conflict-free coloring of graphs with respect to paths is a
special case of conflict-free colorings of hypergraphs, studied in 
Even et al.~\cite{ELRS03jo} and Smorodinsky~\cite{Sm03}.
One of the applications of conflict-free colorings is that it 
represents a frequency assignment for cellular
networks. A cellular network consists of two kinds of nodes:
\emph{base stations} and \emph{mobile agents}. Base stations have
fixed positions and provide the backbone of the network; they are
represented by vertices in \(V\). Mobile agents are the clients of the
network and they are served by base stations. This is done as
follows: Every base station has a fixed frequency; this is represented 
by the coloring \(C\), i.e., colors represent frequencies. If an
agent wants to establish a link with a base station it has to tune
itself to this base station's frequency. Since agents are mobile,
they can be in the range of many different base stations. To avoid
interference, the system must assign frequencies to base stations
in the following way: For any range, there must be a base station
in the range with a frequency that is not used by some other
base station in the range. One can solve the problem by assigning
\(n\) different frequencies to the \(n\) base stations. However,
using many frequencies is expensive, and therefore, a scheme that
reuses frequencies, where possible, is preferable. Conflict-free coloring
problems have been the subject of many recent papers due to their
practical and theoretical interest (see e.g.\
\cite{PT03,HS03,Chenetal2006,EM05,%
BCS2008talg}).
Most approaches in the conflict-free coloring literature
use unique-maximum colorings (a notable exception is 
the `triples' algorithm in \cite{BCS2008talg}), 
because unique-maximum colorings are
easier to argue about in proofs, due to their additional
structure. Another advantage of unique-maximum colorings 
is the simplicity of computing the unique color
in any range (it is always the maximum color), given a
unique-maximum coloring, which can be helpful if very simple
mobile devices are used by the agents.

For general graphs, finding the exact unique-maximum chromatic number of a
graph is NP-complete \cite{Pothen1988TR,LTT1989dam} 
and there is a 
polynomial time $O(\log^2{n})$ approximation algorithm
\cite{DBLP:journals/jal/BodlaenderGHK95}, where $n$ is the number
of vertices. 
Since the problem is
hard in general, it makes sense to study specific graphs.

The  $m\times m$ \emph{grid}, $G_m$, is the {\em cartesian
product} of two paths, each of length $m-1$, that is,
the vertex set of $G_m$ is
$\{0,\ldots,m-1\}\times \{0,\ldots,m-1\}$
and the edges are $\{ \{(x_1,y_1),(x_2,y_2)\} \mid |x_1-x_2|+|y_1-y_2|\le 1\}$.
It is known \cite{orderedcoloring} that for
general planar graphs the unique-maximum chromatic number
is $O(\sqrt{n})$. 
Grid graphs are planar and therefore the $O(\sqrt{n})$ bound applies. 
One might expect that,
since the grid has a  relatively simple and
regular structure, it should not be hard
to calculate its
unique-maximum chromatic number. This is why it is rather striking that, even
though it is not hard to show upper and lower bounds that are only a
small constant multiplicative factor apart, the \emph{exact} value
of these chromatic numbers is not known, and has been the subject
of \cite{BCLMZ2009sirocconopages,BCLMZ2009journal}. 

\paragraph{Paper organization.} 
In the rest of this section we
provide the necessary definitions and some earlier results.
In section~\ref{sec:conpcheckcf}, we prove that it is
coNP-complete to decide whether a given vertex
coloring of a graph is conflict-free with respect to paths.
In section~\ref{sec:ratiogeneral}, we show that for every graph
$\chium(G) \leq 2^{\chicf(G)}-1$ and provide a sequence of graphs for
which the ratio $\chium(G)/\chicf(G)$ tends to $2$.
In section~\ref{sec:gridcfum}, we introduce two games on graphs
that help us relate the two chromatic numbers 
for the square grid graph. In section~\ref{sec:gridlbum}, we show
a lower bound on the unique-maximum chromatic number of the square
grid graph, improving previous results.
Conclusions and open
problems are presented in section~\ref{sec:conclusion}.

\subsection{Preliminaries}

%Note that definition~\ref{def:ordcol} is not the
%typical definition found in the literature. Instead the more
%standard definition is the following. 
%
%\begin{definition}
%A \emph{unique-maximum} $k$-coloring (with respect to paths) 
%of a graph $G$ is a function $C\colon
%V(G) \to \{1,\dots,k\}$ such that for every pair of distinct
%vertices $v$, $v'$, and every path $p$ from $v$ to $v'$, if $C(v) =
%C(v')$, there is an internal vertex $v''$ of $p$ such that $C(v) <
%C(v'')$.
%\end{definition}
%
%It is not hard to show that the two definitions are equivalent
%(see, for example, \cite{orderedcoloring}).

\begin{definition}\label{defn:minorinducedcoloring}
A graph $X$ is a \emph{minor} of $Y$, denoted as $X \isminorof Y$,
if $X$ can be obtained from $Y$ by a sequence of the following three
operations: vertex deletion, edge deletion, and edge contraction.
Edge contraction is the process of merging both endpoints of an edge into
a new vertex, which is connected to all vertices adjacent to the two
endpoints.
Given a unique-maximum coloring $C$ of $Y$, 
we get the {\em induced coloring} of $X$ as follows.
Take a sequence of vertex deletions, edge deletions, and edge contractions
so that we obtain $X$ from $Y$. For the vertex and edge deletion operations, just
keep the colors of the remaining vertices.
For the edge contraction operation, say along edge $xy$,
which gives rise to the new vertex $v_{xy}$, set $C'(v_{xy}) =
\max(C(x),C(y))$, and keep the colors of all other vertices.
\end{definition}

%The following has been proved in \cite{rankings_of_graphs_1998}.
\begin{proposition}\label{prop:minorchio} {\rm \cite{rankings_of_graphs_1998}}
If $X \isminorof Y$, and $C$ is a  unique-maximum coloring of $Y$,
then the induced coloring $C'$ is a unique-maximum coloring of $X$.
Consequently,
$\chium(X) \leq \chium(Y)$.
\end{proposition}

%The above proof implies a specific process to transform a coloring
%of a graph $G$ when applying one of the graph minor operations,
%so that the coloring remains unique-maximum for the resulting 
%graph: if you delete a vertex or an edge, leave the coloring as it is
%in the remaining vertices,
%and if you contract an edge give to the resulting
%vertex the maximum color of the two endpoints of the edge and
%leave the rest of the coloring as it is in the remaining vertices. 
%Here is a lemma that will be useful in proving lower bounds on
%the unique-maximum chromatic number.

%\begin{lemma}\label{lemma:operationscoloring}
%Assume $C$ is an optimal 
%unique-maximum coloring of $G$ and after a sequence of graph
%minor operations you get coloring $C'$ of $G'$.  If $C'$ is using
%$x$ colors less than $C$, then $\chium(G) \geq \chium(G')+x$.
%\end{lemma}
%\begin{proof}
%Since $C'$ is a unique-maximum coloring of $G'$, we have
%$\chium(G') \leq \chium(G) - x$.
%\end{proof}

The (traditional) chromatic number of a graph is denoted by
$\chi(G)$ and is the smallest number of colors in a vertex
coloring for which adjacent vertices are assigned different 
colors. A simple relation between the chromatic numbers we have
defined so far is the following.

\begin{proposition}\label{prop:comparechis}%
For every graph $G$, $\chi(G) \leq \chicf(G) \leq \chium(G)$.
\end{proposition}
\begin{proof}
Since every unique-maximum coloring is also a conflict-free coloring,
we have $\chicf(G) \leq \chium(G)$. A traditional coloring can be
defined as a coloring in which paths of length one are
conflict-free. Therefore every conflict-free coloring is
also a traditional coloring and thus $\chi(G) \leq \chicf(G)$.
\end{proof}

Moreover, we prove that both conflict-free and unique-maximum chromatic
numbers are monotone under taking subgraphs.

\begin{proposition}\label{prop:subgraphmonotonicity}%
If $X \subgraph Y$, then $\chicf(X) \leq \chicf(Y)$ 
and $\chium(X) \leq \chium(Y)$.
\end{proposition}
\begin{proof}
Take the restriction of any
conflict-free or unique-maximum coloring of graph $Y$ to the vertex set
$V(X)$. This is a conflict-free or unique maximum 
coloring of graph $X$, respectively,
because the set of paths of graph $X$ is a subset of all paths of
$Y$.
\end{proof}

% $G - v$ and $G - S$ notations

If $v$ is a vertex (resp. $S$ is a set of vertices)
of graph $G=(V,E)$, denote by $G-v$ (resp. $G-S$)
the graph obtained from $G$ by deleting
vertex $v$ (resp. vertices of $S$) and adjacent edges.

% separators here
\begin{definition}
A subset $S \issubsetof V$ is a \emph{separator} of a connected 
graph $G=(V,E)$ if $G-S$ is disconnected or empty. A separator $S$ is
\emph{minimal} if no proper subset $S' \isstrictsubsetof S$ 
is a separator.
\end{definition}

\section{Deciding whether a coloring is conflict-free}%
\label{sec:conpcheckcf}

In this section, we show a difference between
the two chromatic numbers $\chium$ and $\chicf$, from the
computational complexity aspect. For the notions of complexity
classes, hardness, and completeness, we refer, for example, to
\cite{PapadimitriouCC}.

As we mentioned before, in \cite{Pothen1988TR,LTT1989dam}, it is shown
that
computing $\chium$ for general graphs is NP-complete. To be exact
the following problem is NP-complete: ``Given a graph $G$ and an
integer $k$, is it true that $\chium(G) \leq k$?''.
The above fact implies that it is possible to check in polynomial
time whether a given coloring of a graph is unique-maximum with
respect to paths.
We remark that both the conflict-free and the unique-maximum
properties have to be true in every path of the graph.  
However, a graph with $n$ vertices 
can have exponential in $n$ number of distinct sets of vertices,
each one of which is a vertex set of a path in the graph.
For unique-maximum colorings we can find a shortcut as follows:
Given a (connected) graph $G$ and a vertex coloring of it, 
consider the set of vertices $S$ of unique colors.
Let $u, v\in V\setminus S$ such that they both have the maximum color
that appears in $V\setminus S$.  
If there is a path in $G-S$ from $u$ to $v$, then this path violates the 
unique maximum property. Therefore, 
$S$ has to be a separator in $G$,
which can be checked in polynomial time, otherwise the coloring is
not unique-maximum. If $G-S$ is not empty, we
can proceed analogously for each of its components. 
For conflict-free colorings there is no such shortcut, unless
$\text{coNP}=\text{P}$, as the following theorem implies.

\begin{theorem}
It is coNP-complete to decide whether a given graph and a vertex
coloring of it is conflict-free with respect to paths.
\end{theorem}
\begin{proof}
In order to prove that the problem is coNP-complete, we prove that
it is coNP-hard and also that it belongs to coNP.

We show coNP-hardness by 
a reduction from the complement of the Hamiltonian path problem.
For every graph $G$, we construct in polynomial time 
a graph $G^*$ of polynomial size
together with a coloring $C$ of its vertices such that $G$ has no
Hamiltonian path if and only if $C$ is conflict-free with respect
to paths of $G^*$. 

Assume the vertices of graph $G$ are $v_1$, $v_2$, \ldots, $v_n$.
Then, graph $G^*$ consists of two isomorphic copies of $G$,
denoted by  $\Gup$ and $\Gdown$,
with vertex sets
$\vup_1$, $\vup_2$, \ldots, $\vup_n$ and 
$\vdown_1$, $\vdown_2$, \ldots, $\vdown_n$, respectively. 
Additionally, for every $1\le i\le n$,  $G^*$ contains the path 
$$P_i=\vup_i, v_{i,1}, v_{i,2}, \dots, v_{i,i-1}, v_{i,i+1}, \dots, v_{i,n},
\vdown_i,$$
where, for every $i$, 
$v_{i,1}, v_{i,2}, \dots, v_{i,i-1}, v_{i,i+1}, \dots, v_{i,n}$ 
are new vertices.
We use the following notation for the two possible directions 
to traverse this path:
$$\Pdown_{i} = 
  (v_{i,1},\dots, v_{i,i-1}, v_{i,i+1}, \dots, v_{i,n}),$$
$$\Pup_{i} = 
  (v_{i,n},\dots, v_{i,i+1}, v_{i,i-1}, \dots, v_{i,1}).$$
We call paths $P_i$ \emph{connecting paths}.

We now describe the coloring of $V(G^*)$. For every $i$, we set 
$C(\vup_i) = C(\vdown_i) = i$. For every $i > j$, we set
$C(v_{i,j})=C(v_{j,i})=n+\binom{i-1}{2}+j$.
Observe that every color occurs exactly in two vertices of $G^*$.

If $G$ has a Hamiltonian path, say $v_1v_2 \dotsc v_n$, then there
is a path through all vertices of $G^*$, either
\[\vup_1 \Pdown_1 \vdown_1 \vdown_2 \Pup_2 \vup_2 \dots 
  \vup_{n-1} \Pdown_{n-1} \vdown_{n-1} \vdown_n \Pup_n \vup_n, 
  \text{ if $n$ is even,}
\]
or
\[\vup_1 \Pdown_1 \vdown_1 \vdown_2 \Pup_2 \vup_2 \dots 
  \vdown_{n-1} \Pup_{n-1} \vup_{n-1} \vup_n \Pdown_n \vdown_n, 
  \text{ if $n$ is odd.}
\]
But then, this path has no uniquely occurring color and thus
$C$ is not conflict-free.

Suppose now that $C$ is not a conflict-free coloring. We prove that
$G$ has a Hamiltonian path.

By the assumption, 
there is a path
$P$ in $G^*$ which is not conflict-free. This path must contain
none or both vertices of each color.
Therefore, $P$
can not be completely contained in $\Gup$, or in $\Gdown$, or in
some $P_i$. Also, $P$ can not contain only one of $\vup_i$ and
$\vdown_i$, for some $i$. Therefore, $P$ must contain both
$\vup_i$ and $\vdown_i$ for a non-empty subset of indices $i$. 

Then, it must contain completely some $P_i$, because vertices
in $\Gup$ and $\Gdown$ can only be connected with some complete
$P_i$. But since each one of the $n-1$ colors of this
$P_i$ occurs in a different connecting paths, 
$P$ must contain a vertex in every connecting path.
But then $P$ must
contain every $\vup_i$ and $\vdown_i$, because vertices in
$P_i$ can only be connected to the rest of the graph through one
of $\vup_i$ or $\vdown_i$. 

Suppose that $P$ is not a Hamiltonian path of $G^*$.
Observe that 
if $P$ does not contain all vertices of some connecting path $P_i$, then
one of its end vertices should be there. 
If $P$ does not contain vertex $v_{i,j}$, then it can not contain $v_{j,i}$
either.
But then one end vertex of $P$ should be on $P_i$, the other one on $P_j$, and
all other vertices of $G^*$ are on $P$. Therefore, we can extend $P$ such that
it contains $v_{i,j}$ and $v_{j,i}$ as well. So assume in the sequel that $P$
is a Hamiltonian path of $G^*$.

Now we modify $P$, if necessary, so that 
both of its end-vertices $e$ and $f$ lie
in $V(\Gup)\cup V(\Gdown)$.
If $e$ and $f$ are adjacent in $G^*$, then add the edge $ef$ to $P$ and we get
a Hamiltonian cycle of  $G^*$. Now remove one of its edges which is either in
$\Gup$, or in $\Gdown$
and get the desired Hamiltonian path.
Suppose now that $e$ and $f$ are not adjacent, and $e$ is on one of the
connecting paths. Then $e$ should be adjacent to the end vertex $e'$ of that
connecting path, which is in $\Gup$ or in $\Gdown$.
Add edge $ee'$ to $P$. We get a cycle and a path joined in $e'$. Remove the
other edge of the cycle adjacent to $e'$. We have a Hamiltonian path now,
whose end vertex is $e'$ instead of $e$. 
Proceed analogously for $f$, if necessary.

Now we have a Hamiltonian path $P$ of $G^*$ with end-vertices in $V(\Gup)\cup
V(\Gdown)$. Then, $P$ is in the form, say, 
\[\vup_1 \Pdown_1 \vdown_1 \vdown_2 \Pup_2 \vup_2 \dots 
  \vup_{n-1} \Pdown_{n-1} \vdown_{n-1} \vdown_n \Pup_n \vup_n, 
  \text{ if $n$ is even,}
\]
or
\[\vup_1 \Pdown_1 \vdown_1 \vdown_2 \Pup_2 \vup_2 \dots 
  \vdown_{n-1} \Pup_{n-1} \vup_{n-1} \vup_n \Pdown_n \vdown_n, 
  \text{ if $n$ is odd.}
\]
%because if the first occurrence of $\vup_i$ or $\vdown_i$ is not
%followed by $P_i$, then $p$ must end with some vertex of
%$P_i$, which is impossible.
But then, $v_1v_2\dotsc v_n$ is a 
Hamiltonian path in $G$. 

\smallskip

Finally, the problem is in coNP because one can verify that a
coloring of a given graph is not conflict-free in polynomial time, 
by giving the corresponding path.
\end{proof}

We show an example graph $G$, its
transformation graph $G^*$, and its coloring $C$ in
figure~\ref{fig:graphexmp}.

\def\scaleconpfigures{0.6}%
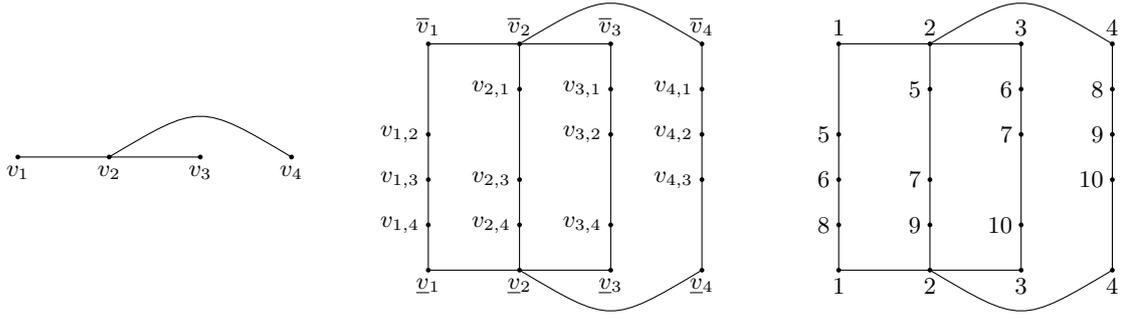
\begin{figure}[htb]
\centering
\begin{tikzpicture}[scale=\scaleconpfigures]
%\small
\footnotesize % smaller size of labels
\begin{scope}[yshift=2.5cm] % Graph G
\coordinate (v1) at (0,0);
\coordinate (v2) at (2,0);
\coordinate (v3) at (4,0);
\coordinate (v4) at (6,0);

\foreach \v in {v1, v2, v3, v4}
   \fill (\v) circle (1.5pt);

\draw (v1)--(v2)--(v3);
\draw (v2) .. controls (4, 1.2) .. (v4);

\node [below] at (v1) {$v_1$};
\node [below] at (v2) {$v_2$};
\node [below] at (v3) {$v_3$};
\node [below] at (v4) {$v_4$};
\end{scope}  % end of graph G
\begin{scope}[xshift=9cm] % graph G^*
\coordinate (vup1) at (0,5);
\coordinate (vup2) at (2,5);
\coordinate (vup3) at (4,5);
\coordinate (vup4) at (6,5);

\coordinate (vdown1) at (0,0);
\coordinate (vdown2) at (2,0);
\coordinate (vdown3) at (4,0);
\coordinate (vdown4) at (6,0);

\coordinate (v12) at (0,3);
\coordinate (v13) at (0,2);
\coordinate (v14) at (0,1);

\coordinate (v21) at (2,4);
\coordinate (v23) at (2,2);
\coordinate (v24) at (2,1);

\coordinate (v31) at (4,4);
\coordinate (v32) at (4,3);
\coordinate (v34) at (4,1);

\coordinate (v41) at (6,4);
\coordinate (v42) at (6,3);
\coordinate (v43) at (6,2);

\foreach \v in {
    vup1, vup2, vup3, vup4, 
    vdown1, vdown2, vdown3, vdown4, 
    v12, v13, v14, v21, v23, v24,
    v31, v32, v34, v41, v42, v43}
   \fill (\v) circle (1.5pt);

\draw (vup1)--(vup2)--(vup3);
\draw (vup2) .. controls (4, 6.2) .. (vup4);

\draw (vdown1)--(vdown2)--(vdown3);
\draw (vdown2) .. controls (4, -1.2) .. (vdown4);

\draw (vup1)--(v12)--(v13)--(v14)--(vdown1);
\draw (vup2)--(v21)--(v23)--(v24)--(vdown2);
\draw (vup3)--(v31)--(v32)--(v34)--(vdown3);
\draw (vup4)--(v41)--(v42)--(v43)--(vdown4);

\node [above] at (vup1) {$\vup_1$};
\node [above] at (vup2) {$\vup_2$};
\node [above] at (vup3) {$\vup_3$};
\node [above] at (vup4) {$\vup_4$};

\node [below] at (vdown1) {$\vdown_1$};
\node [below] at (vdown2) {$\vdown_2$};
\node [below] at (vdown3) {$\vdown_3$};
\node [below] at (vdown4) {$\vdown_4$};

\node [left] at (v12) {$v_{1,2}$};
\node [left] at (v13) {$v_{1,3}$};
\node [left] at (v14) {$v_{1,4}$};

\node [left] at (v21) {$v_{2,1}$};
\node [left] at (v23) {$v_{2,3}$};
\node [left] at (v24) {$v_{2,4}$};

\node [left] at (v31) {$v_{3,1}$};
\node [left] at (v32) {$v_{3,2}$};
\node [left] at (v34) {$v_{3,4}$};

\node [left] at (v41) {$v_{4,1}$};
\node [left] at (v42) {$v_{4,2}$};
\node [left] at (v43) {$v_{4,3}$};
\end{scope} % end of graph G^*
\begin{scope}[xshift=18cm] % coloring C
\coordinate (vup1) at (0,5);
\coordinate (vup2) at (2,5);
\coordinate (vup3) at (4,5);
\coordinate (vup4) at (6,5);

\coordinate (vdown1) at (0,0);
\coordinate (vdown2) at (2,0);
\coordinate (vdown3) at (4,0);
\coordinate (vdown4) at (6,0);

\coordinate (v12) at (0,3);
\coordinate (v13) at (0,2);
\coordinate (v14) at (0,1);

\coordinate (v21) at (2,4);
\coordinate (v23) at (2,2);
\coordinate (v24) at (2,1);

\coordinate (v31) at (4,4);
\coordinate (v32) at (4,3);
\coordinate (v34) at (4,1);

\coordinate (v41) at (6,4);
\coordinate (v42) at (6,3);
\coordinate (v43) at (6,2);

\foreach \v in {
    vup1, vup2, vup3, vup4, 
    vdown1, vdown2, vdown3, vdown4, 
    v12, v13, v14, v21, v23, v24,
    v31, v32, v34, v41, v42, v43}
   \fill (\v) circle (1.5pt);

\draw (vup1)--(vup2)--(vup3);
\draw (vup2) .. controls (4, 6.2) .. (vup4);

\draw (vdown1)--(vdown2)--(vdown3);
\draw (vdown2) .. controls (4, -1.2) .. (vdown4);

\draw (vup1)--(v12)--(v13)--(v14)--(vdown1);
\draw (vup2)--(v21)--(v23)--(v24)--(vdown2);
\draw (vup3)--(v31)--(v32)--(v34)--(vdown3);
\draw (vup4)--(v41)--(v42)--(v43)--(vdown4);

\node [above] at (vup1) {$1$};
\node [above] at (vup2) {$2$};
\node [above] at (vup3) {$3$};
\node [above] at (vup4) {$4$};

\node [below] at (vdown1) {$1$};
\node [below] at (vdown2) {$2$};
\node [below] at (vdown3) {$3$};
\node [below] at (vdown4) {$4$};

\node [left] at (v12) {$5$};
\node [left] at (v13) {$6$};
\node [left] at (v14) {$8$};

\node [left] at (v21) {$5$};
\node [left] at (v23) {$7$};
\node [left] at (v24) {$9$};

\node [left] at (v31) {$6$};
\node [left] at (v32) {$7$};
\node [left] at (v34) {$10$};

\node [left] at (v41) {$8$};
\node [left] at (v42) {$9$};
\node [left] at (v43) {$10$};
\end{scope} % end of coloring C
\end{tikzpicture}
\caption{Example graphs $G$, $G^*$, and coloring $C$ of $G^*$}
\label{fig:graphexmp}
\end{figure}

\section{The two chromatic numbers of general graphs}%
\label{sec:ratiogeneral}

We have seen that $\chium(G) \geq \chicf(G)$
(proposition~\ref{prop:comparechis}). 
In this section we show that $\chium(G)$ can not be larger
than an exponential function of $\chicf(G)$. We also provide an
infinite
sequence of graphs $H_1$, $H_2$, \ldots, for which
$\lim_{k\to\infty}(\chium(H_k)/\chicf(H_k)) = 2$.

The path of $n$ vertices is denoted by $P_n$.
It is known that $\chium(P_n) = \floor{\log_{2}{n}}+1$ (see for
example \cite{ELRS03jo}). 

\begin{lemma}\label{lemma:pathcf}
For every path $P_n$, $\chicf(P_n) = \floor{\log_2{n}}+1$.
\end{lemma}
\begin{proof}
By proposition~\ref{prop:comparechis}, $\chicf(P_n) \leq
\chium(P_n)$. We prove a matching lower bound by induction.
We have $\chicf(P_1) \geq 1$. For $n > 1$, there is a uniquely
occurring color in any conflict-free coloring of the the whole path 
$P_n$. Then, $\chicf(P_n) \geq 1 + \chicf(P_{\floor{n/2}})$, which 
implies $\chicf(P_n) \geq \floor{\log_2{n}}+1$.
\end{proof}

Moreover, we are going to use the following result (lemma 5.1 of
\cite{orderedcoloring}): If the longest path of $G$ has $k$ vertices, 
then $\chium(G) \leq k$.

\begin{proposition}
For every graph $G$, $\chium(G) \leq 2^{\chicf(G)}-1$.
\end{proposition}
\begin{proof}
Set $j=\chicf(G)$. For any path $P\issubgraphof G$,
$\chicf(P)\le j$, therefore, by lemma~\ref{lemma:pathcf}, 
the longest path has at most $2^{j}-1$ vertices,
so by lemma~5.1 of \cite{orderedcoloring}, $\chium(G) \leq 2^j-1$.
\end{proof}

We define recursively the following sequence of graphs:
Graph $H_0$ is a single vertex. Suppose that we have already defined
$H_{k-1}$. Then $H_k$ consists of 
%
%(a)~a $K_{2\chicf(H_{k-1})+1}$, 
%i.e., a clique of $2\chicf(H_{k-1})+1$ vertices, 
%(b)~$2\chicf(H_{k-1})+1$
%isomorphic copies of $H_{k-1}$, and 
%(c)~$2\chicf(H_{k-1})+1$ edges, each one of them connecting a
%different vertex of $K_{2\chicf(H_{k-1})+1}$ with any vertex of a
%different copy of $H_{k-1}$.
%
(a) a $K_{2^{k+1}-1}$, (b) 
${2^{k+1}-1}$ copies of $H_{k-1}$, and (c) 
for for each $i$, $1\le i\le 2^{k+1}-1$, the $i$-th vertex of the 
$K_{2^{k+1}-1}$ is connected by an edge 
to one of the vertices of the $i$-th copy of
$H_{k-1}$.

%\begin{comment}
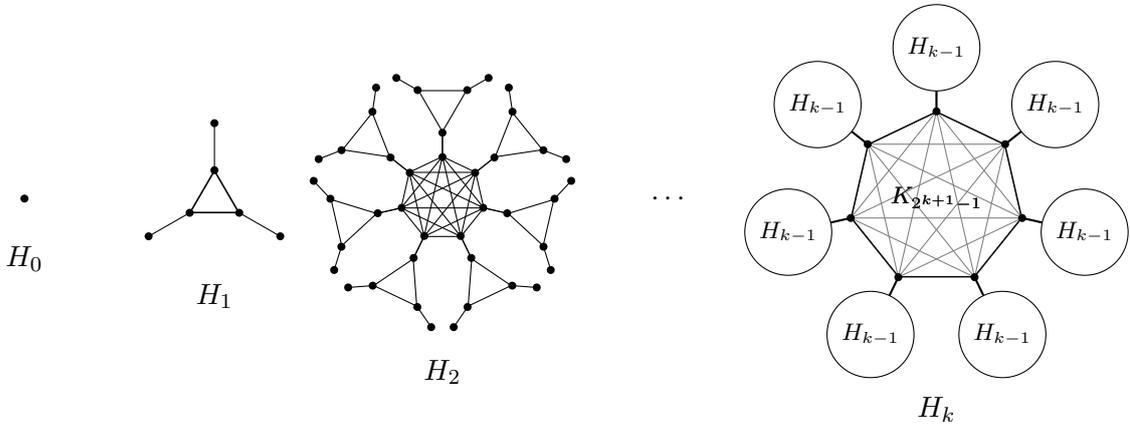
\begin{figure}[htb]
\centering
\begin{tikzpicture}
\newcommand*{\vertexsize}{1.5pt}%
%
% H_0
\begin{scope}
\fill (0,0) circle (\vertexsize);
\node [below] at (0,-0.5) {$H_0$};
\end{scope}
%
% H_1
\begin{scope}[xshift=+2.5cm]
\node [regular polygon, regular polygon sides = 3, name = tri,
       draw] 
  at (0,0) {};
\foreach \i in {1,...,2} {
   \foreach \j in {\i, ..., 3} {
      \ifnum \i < \j
         \draw (tri.corner \i) -- (tri.corner \j);
      \fi
   }
}
\foreach \i in {1, ..., 3}
{
   \draw 
      (tri.corner \i) -- (-30+\i*120:1cm);
   \fill (-30+\i*120:1cm) circle (\vertexsize);
   \fill (tri.corner \i) circle (\vertexsize);
}
\node [below] at (0,-1) {$H_1$};
\end{scope}
%
% H_2
\begin{scope}[xshift=+5.5cm]
\newcommand*{\awaylen}{1.1cm}
\coordinate (center) at (0,0);
%
%
% clique
%\draw [fill=white] (0,0) circle (20pt);
\node [regular polygon, regular polygon sides = 7, name = clique,
       minimum size = 1.1cm] 
  at (center) { };
\foreach \i in {1,...,6} {
   \foreach \j in {\i, ..., 7} {
      \ifnum \i < \j
         \draw (clique.corner \i) -- (clique.corner \j); %
      \fi
   }
}
%\node [regular polygon, regular polygon sides = 7, 
%       draw] 
%  at (center) {$K_{7}$};
%
%\foreach \i in {1, ..., 7} 
%   \fill (clique.corner \i) circle (\vertexsize);
%
%
%
% edges
\foreach \i in {1, ..., 7}
{
   \draw (clique.corner \i) -- (38.57142857+\i*51.42857:\awaylen);
}%
%
%
% isomorphic copies
\newcommand*{\moreawaylen}{1.25cm}%
\newcommand*{\smallawaylen}{0.7cm}%
\foreach \i in {1, ..., 7}
{
   \begin{scope}[
     shift={(38.57142857+\i*51.42857:\moreawaylen)}]
   \begin{scope}
       [rotate = 180-51.42857+\i*51.42857]
      \node [regular polygon, regular polygon sides = 3, 
             rotate=180-51.42857+\i*51.42857, 
             draw, fill=white, name = tric] 
         at (0,0) {};
      \foreach \j in {1, ..., 3}
      {
         \fill (-30+\j*120:\smallawaylen) circle (\vertexsize);
         \fill (tric.corner \j) circle (\vertexsize);
         \draw (tric.corner \j) -- (-30+\j*120:\smallawaylen);
      }
   \end{scope}
   \end{scope}
}
\node [below] at (0,-2) {$H_2$};
\end{scope}
\begin{scope}[xshift=+8.5cm]
\node at (0,0) {$\dotsm$};
\end{scope}
%
% H_k
\begin{scope}[xshift=+12.0cm]
\newcommand*{\awaylen}{2cm}
\coordinate (center) at (0,0);
%
%
% clique
%\draw [fill=white] (0,0) circle (20pt);
\node [regular polygon, regular polygon sides = 7, name = clique,
       draw] 
  at (center) {\footnotesize $K_{2^{k+1}-1}$};
\foreach \i in {1,...,6} {
   \foreach \j in {\i, ..., 7} {
      \ifnum \i < \j
         \draw [gray] (clique.corner \i) -- (clique.corner \j);
      \fi
   }
}
\node [regular polygon, regular polygon sides = 7, 
       draw] 
  at (center) {\footnotesize $K_{2^{k+1}-1}$};
\foreach \i in {1, ..., 7} 
   \fill (clique.corner \i) circle (\vertexsize);
%
%
%
% edges
\foreach \i in {1, ..., 7}
{
   \draw [thick] 
      (clique.corner \i) -- (38.57142857+\i*51.42857:\awaylen);
}
%
%
% isomorphic copies
\foreach \i in {1, ..., 7}
{
   \node [circle, draw, fill=white] 
      at (38.57142857+\i*51.42857:\awaylen) {\footnotesize $H_{k-1}$};
}
\node [below] at (0,-2.5) {$H_k$};
\end{scope}
\end{tikzpicture}
\caption{Sequence of graphs}
\label{fig:H_k_sequence}
\end{figure}
%\end{comment}

\begin{lemma}\label{lemma:hedgehogcf}
For $k\ge 0$, $\chicf(H_k) = 2^{k+1}-1$.
\end{lemma}
\begin{proof}
By induction on $k$. For $k=0$, $\chicf(H_0) = 1$. For $k>0$, we have
$H_k \containssubgraph K_{2^{k+1}-1}$, therefore, 
$\chicf(H_k) \geq 2^{k+1}-1$.

In order to prove that $\chicf(H_k) \leq 2^{k+1}-1$,
it is enough to describe a conflict-free coloring of $H_k$ with
$2^{k+1}-1$ colors, given a conflict-free coloring of $H_{k-1}$
with $2^k-1$ colors:
We color the vertices of the clique $K_{2^{k+1}-1}$
with colors $1, 2, \ldots , 2^{k+1}-1$ such that the $i$-th vertex is colored
with color $i$.
Consider these colors mod $2^{k+1}-1$, e. g. color $2^{k+1}$ is identical to
color 1. 
Recall that the $i$-th copy of $H_{k-1}$ has a vertex connected to the $i$-th
vertex 
of $K_{2^{k+1}-1}$, and by induction we know that 
$\chicf(H_{k-1})=2^{k}-1$.
Color the $i$-th copy of $H_{k-1}$, with colors $i+1, i+2, \ldots , i+2^{k}-1$.

We claim that this vertex coloring of $H_k$ is conflict-free.
If a path is completely contained in a copy of
$H_{k-1}$, then it is conflict-free by induction.
If a path is completely contained in the clique $K_{2^{k+1}-1}$, then it is
also conflict-free, because all colors in the clique part are
different. If a path contains vertices from a single 
copy of $H_{k-1}$, say, the $i$-th copy, and the clique, then the $i$-th 
vertex of the
clique is on the path and uniquely colored.
The last case is when a path contains vertices from exactly two
copies of $H_{k-1}$. Suppose that these are the $i$-th and $j$-th copies of 
 $H_{k-1}$, $1\le i<j\le 2^{k+1}-1$. 
If $i+2^k-1<j$, then color $j$ is unique in the path; indeed, 
the $i$-th copy of $H_{k-1}$ is colored with colors $i+1, \ldots i+2^k-1$, 
and the $j$-th copy of $H_{k-1}$ is colored with colors $j+1, \ldots j+2^k-1$, 
while color $j$ appears only once in $K_{2^{k+1}-1}$. 
Similarly, if $i+2^k-1\ge j$, then color $i$ is unique in the path.
\end{proof}

\begin{lemma}\label{lemma:hedgehogumub}
$\chium(H_k) \leq 2^{k+2}-k-3$.
\end{lemma}
\begin{proof}
By induction. For $k=0$, $\chium(H_0)=1$. 
For $k>0$, 
in order to color $H_k$
use the $2^{k+1}-1$ different highest 
colors for the clique part.
By the inductive hypothesis $\chium(H_{k-1}) \leq 2^{k+1}-k-2$.
For each copy of $H_{k-1}$, 
use the same coloring with the $2^{k+1}-k-2$ lowest
colors. 
This coloring of $H_k$ is unique maximum. Indeed, if a path is
contained in a copy of $H_{k-1}$ then it is unique maximum
by induction, and if it contains a vertex
in the clique part, then it is also unique maximum.
The total number of colors is $2^{k+2}-k-3$.
\end{proof}

\begin{lemma}\label{lemma:cliqueplusgraph}
If $Y$ is a graph that consists of a $K_\ell$ and $\ell$
isomorphic copies of a connected graph $X$, such that 
for $1\le i\le\ell$ a vertex of 
it $i$-th copy is connected to the $i$-th vertex of $K_\ell$ by an edge.
Then we have 
$\chium(Y) \geq \ell - 1 + \chium(X)$
\end{lemma}
\begin{proof}
By induction on $\ell$. For $\ell=1$, we have that $\chium(Y) \geq
\chium(X)$, because $Y \containssubgraph X$. For the inductive
step, for $\ell > 1$, if $Y$ consists of  a $K_{l}$ and $\ell$
copies of $X$, then $Y$ is connected, and thus contains a vertex
$v$ with unique color. But then, $Y-v \containssubgraph Y'$, where
$Y'$ is a graph that consists of a $K_{\ell-1}$ and ${\ell-1}$
isomorphic copies of a $X$, each connected to a different vertex of
$K_{\ell-1}$, 
and thus $\chium(Y) = 1 + \chium(Y') 
\geq \ell-1+\chium(X)$.
\end{proof}

\begin{lemma}\label{lemma:hedgehogumlb}
$\chium(H_k) \geq 2^{k+2}-2k-3$.
\end{lemma}
\begin{proof}
By induction. For $k=0$, $\chium(H_0)=1$. 
For $k >0$, by the inductive hypothesis and lemma~\ref{lemma:cliqueplusgraph},
$\chium(H_k) \geq 2^{k+1}-1-1+2^{k+1}-2(k-1)-3 = 2^{k+2}-2k-3$
\end{proof}

\begin{theorem}
We have $\lim_{k\to\infty}(\chium(H_k)/\chicf(H_k))=2$.
\end{theorem}
\begin{proof} From lemmas \ref{lemma:hedgehogcf}, \ref{lemma:hedgehogumub},
\ref{lemma:hedgehogumlb}, we have 
\[
  \frac{2^{k+2}-2k-3}{2^{k+1}-1} \leq 
  \frac{\chium(H_k)}{\chicf(H_k)} \leq
  \frac{2^{k+2}-k-3}{2^{k+1}-1}
\]
which implies that the ratio tends to $2$.
\end{proof}

%------------------------------------------------

\section{The two chromatic numbers of a square grid}
\label{sec:gridcfum}

In this section, we define two games on graphs, each played by two
players. The first game characterizes completely the 
unique-maximum chromatic number of the graph.
The second game is related to the conflict-free chromatic number
of the graph.
We use the two games to prove that the conflict-free chromatic
number of the square grid is a function of the unique-maximum
chromatic number of the square grid. This is useful because it
allows to translate existing lower bounds on the unique-maximum chromatic
number of the square grid to lower bounds on the 
corresponding conflict-free
chromatic number.
For any graph $G$, and subset of its vertices $V'\subset V(G)$,
let $G[V']$ denote the subgraph of $G$ induced by $V'$.

The first game (which is played on a graph $G$ by two players) 
is the \emph{connected component game}: 

\begin{lstlisting}
$i \gets 0$; $G^0 \gets G$
while $V(G^i) \neq \emptyset$:
   increment $i$ by $1$
   Player 1 chooses a connected component $S^i$ of $G^{i-1}$ 
   Player 2 chooses a vertex $v_i \in S^i$
   $G^i \gets G^{i-1}[S^i \setminus \{v_i\}]$ 
\end{lstlisting}

The game is finite, because if $G^i$ is not empty, then $G^{i+1}$
is a strict subgraph of $G^i$.
The result of the game is its length, that is, the final value of $i$.
Player 1 tries to make the final value of $i$ as large as possible
and thus is the maximizer player. 
Player 2 tries to make the final value of $i$ as small as possible
and thus is the minimizer player.
If both players play optimally, then the result is the
\emph{value} of the connected component game on graph $G$, 
which is denoted by $\vcs(G)$.

\begin{proposition}
In the connected component game, there is a strategy for player 2
(the minimizer), 
so that the result of the game is at most $\chium(G)$, i.e., 
$\vcs(G) \leq \chium(G)$.
\end{proposition}
\begin{proof}
By induction on $\chium(G)$: If $\chium(G) = 0$, i.e., the graph is
empty, the value of the game is 0. If $\chium(G) = k > 0$,
then in the first turn some connected component $S_1$ is chosen by player 1. 
Then, the strategy of player 2 is to take an optimal 
unique-maximum coloring $C$ of
$G$ and choose a vertex $v_1$ in $S^1$ that has a unique color in
$S^1$. Then, 
$G^1 = G[S^1\setminus \{v_1\}] \propersubgraph G^0$ 
and the restriction of
$C$ to $S^1 \setminus\{v_1\}$ is a unique-maximum coloring 
of $G^1$
that is using at most $k-1$
colors. Thus, $\chium(G^1) \leq k-1$, and by the inductive hypothesis
player 2 has a strategy so that the result of the game on $G^1$ is at most
$k-1$.
Therefore, player 2 has a strategy so that the result of the game
on $G^0=G$ is at most $1+k-1 = k$.
\end{proof}

\begin{lemma}\label{lem:gminuschium}
For every $v \in V(G)$, $\chium(G-v) \geq \chium(G) - 1$
\end{lemma}
\begin{proof}
Assume for the sake of contradiction that there exists a $v \in
V(G)$ for which $\chium(G-v) < \chium(G) - 1$. Then an optimal
coloring of $G-v$ can be extended to a coloring of $G$, where $v$
has a new unique maximum color. Therefore there is a coloring of $G$ that
uses less than $\chium(G) - 1 + 1 = \chium(G)$ colors; a
contradiction.
\end{proof}

\begin{proposition}\label{prop:cscfplayer1}
In the connected component game, there is a strategy for player 1
(the maximizer), so that the result of the game is at least
$\chium(G)$, i.e., $\vcs(G) \geq \chium(G)$.
\end{proposition}
\begin{proof}
By induction on $\chium(G)$: If $\chium(G) = 0$, i.e., the graph is
empty, the result of the game is zero. If $\chium(G) = k > 0$,
the strategy of player 1 is to choose a connected component $S^1$ such that
$\chium(G[S^1]) = k$. 
%there is always such a connected component,
%because otherwise one could color all connected components of $G$
%with less than $k$ colors (a contradiction to $\chium(G)=k$). 
For
every choice of $v_1$ by Player 2, by lemma~\ref{lem:gminuschium},
$\chium(G^1) \geq k-1$, and thus, by the inductive hypothesis player
1 has a strategy so that the result of the game on $G^1$ is at
least $k-1$.
Therefore, the result of the game on $G^0=G$ is at least $1+k-1 =
k$.
\end{proof}

\begin{corollary}
For every graph, $\vcs(G) = \chium(G)$.
\end{corollary}

The second game 
(also played on a graph $G$ by two players) 
is the \emph{path game}:

\begin{lstlisting}
$i \gets 0$; $G^0 \gets G$
while $V(G^i) \neq \emptyset$:
   increment $i$ by $1$
   Player 1 chooses the set of vertices $S^i$ of a path of $G^{i-1}$
   Player 2 chooses a vertex $v_i \in S^i$
   $G^i \gets G^{i-1}[S^i \setminus \{v_i\}]$ 
\end{lstlisting}

The only difference with the connected component game is that in the path
game the vertex set $S^i$ that maximizer chooses is the vertex set
of a path of the graph $G^{i-1}$.
If both players play optimally, then the result is the
\emph{value} of the path game on graph $G$, which is denoted by
$\vp(G)$.

\begin{proposition}\label{prop:vpcf}
In the path game, there is a strategy for player 2 (the minimizer), 
so that the result of the game is at most $\chicf(G)$, i.e., 
$\vp(G) \leq \chicf(G)$.
\end{proposition}
\begin{proof}
By induction on $\chicf(G)$: If $\chicf(G) = 0$, i.e., the graph is
empty, the value of the game is 0. If $\chicf(G) = k > 0$,
then in the first turn some vertex set $S^1$ of a path of $G$
is chosen by player 1. 
Then, the strategy of player 2 is to find an optimal 
conflict-free coloring $C$ of
$G$ and choose a vertex $v_1$ in $S^1$ that has a unique color in
$S^1$. Then, 
$G^1 = G[S^1\setminus \{v_1\}] \propersubgraph G^0$ 
and the restriction of
$C$ to $S^1 \setminus\{v_1\}$ is a conflict-free coloring 
of $G^1$
that is using at most $k-1$
colors. Thus, $\chicf(G^1) \leq k-1$, and by the inductive hypothesis
player 2 has a strategy so that the result of the game is at most
$k-1$.
Therefore, player 2 has a strategy so that the result of the game is at most
$1+k-1 = k$.
\end{proof}

A proposition analogous to \ref{prop:cscfplayer1} for the
path game is not true. For example, for the complete binary tree
of four levels (with 15 vertices, 8 of which are leaves), $B_4$, 
it is not difficult to check that $\vp(B_4) = \vp(P_7) = 3$, but
$\chicf(B_4)=4$.

Now, we are going to concentrate on the square grid graph. Assume
that $m$ is even. We
intend to translate a strategy of player 1 (the maximizer) on the 
connected component game for graph $G_{m/2}$ to a strategy for
player on the path game for graph $G_m$. 

%#2

Observe that for every connected graph $G$, 
there is an ordering of its vertices, $v_1, v_2, \ldots , v_n$
such that the subgraph induced by the first $k$
vertices (for every $1\le k\le n$) is also connected. 
Just pick a vertex to be $v_1$, and add the other vertices one by one 
such that the new vertex $v_i$ is connected to the graph induced by 
$v_1, \ldots , v_{i-1}$. This is possible, since $G$ itself is connected.
We call such an ordering of the vertices an \emph{always-connected ordering}.

%
%\begin{lemma}\label{lemma:alwaysconnected}
%For every connected set of vertices $S$ in $G_{m/2}$, there is an
%ordering of its vertices $v_1$, $v_2$, \ldots, $v_k$, such that
%for every $j$ with $1\leq j\leq k$ 
%the set $S_j = \{v_r \mid 1 \leq r \leq j\}$ is connected
%in $G_{m/2}$. We call this an always-connected ordering of $S$.
%\end{lemma}
%\begin{proof}
%Construct the ordering as follows. Pick any vertex in $S$
%to be the initial vertex in the ordering.
%It is always possible to extend the ordering by picking a vertex
%of $S$ which is connected to at least one vertex already picked.
%If no such vertex exists, then $S$ is not connected, which is a
%contradiction.
%\end{proof}
%

Now we decompose the vertex set of $G_m$ into groups of four vertices, 
$$Q_{x,y} = \{(2x,2y),(2x+1,2y),(2x,2y+1),(2x+1,2y+1)\}, $$ 
for $0\le x, y < m/2$,
called {\em special quadruples}, or briefly quadruples.
We denote the set of quadruples with 
$W_m = \left\{ Q_{x,y} \mid 0 \le x, y < m/2 \right\}$ and let 
$\tau(x,y) = Q_{x, y}$ be a bijection between 
vertices of $V(G_{m/2})$ and $W_m$.
Extend $\tau$ for subsets of vertices of $G_{m/2}$ in a natural way,
for any $S \issubsetof V(G_{m/2})$, 
%$\tau(S)=\left\{  \tau(x, y) \mid (x, y)\in S \right\}$.
$\tau(S)=\bigcup_{(x,y) \in S} \tau(x,y)$.
Define also a kind of inverse $\tau'$ of $\tau$ as
$\tau'(x, y) = (\floor{x/2},\floor{y/2})$ for any $0\le x, y< m$, 
and for any $S \issubsetof V(G_{m})$, 
$\tau'(S) = \left\{ \tau'(x, y) \mid (x, y)\in S \right\}$.

%#3

Let $(x,y) \in V(G_{m/2})$. We call vertices 
$(x, y+1)$, $(x, y-1)$, $(x-1, y)$, and $(x+1, y)$,
if they exist,  the {\em upper}, {\em lower}, {\em left}, and {\em right
neighbors} of $(x,y)$, respectively.
Similarly, quadruples $Q_{x, y+1}$, $Q_{x, y-1}$, $Q_{x-1, y}$, 
and $Q_{x+1, y}$
the {\em upper}, {\em lower}, {\em left}, and {\em right
neighbors} of $Q_{x,y}$, respectively.

Quadruple $Q_{x,y}$ induces four edges in $G_m$,
$\{(2x+1, 2y),(2x+1, 2y+1)\}$,
$\{(2x, 2y),(2x, 2y+1)\}$,
$\{(2x, 2y),(2x, 2y+1)\}$,
$\{(2x+1, 2y),(2x+1, 2y+1)\}$, 
we call them {\em upper}, {\em lower}, {\em left}, and {\em right} 
edges of $Q_{x,y}$.

%For every direction $d$, we define its mirror direction $\dmirror$
%as follows: 
%$\mirror{E}=W$,
%$\mirror{N}=S$,
%$\mirror{W}=E$,
%$\mirror{S}=N$,
%i.e., directions are mirrored along the 
%$E$-$W$
%and
%$N$-$S$
%axes.
By {\em direction} $d$, we mean one of the four basic directions,
\emph{up}, \emph{down}, \emph{left}, \emph{right}.
For a given set $S \issubsetof V(G_{m/2})$, we say that $v \in S$ is
\emph{open} in $S$ in direction $d$, if its neighbor in direction
$d$ is not in $S$. In this case we also say that $\tau(v)$ is open 
in $\tau(S)$ in direction $d$.

\begin{lemma}\label{lemma:pathspanned}
If $S$ induces a connected subgraph 
in $G_{m/2}$, then there is a path in $G_{m}$
whose vertex set is $\tau(S)$.
\end{lemma}
\begin{proof}
We prove a stronger statement: 
If $S$ induces a connected subgraph 
in $G_{m/2}$, then 
there is a cycle $C$ in $G_{m}$ whose 
vertex set is $\tau(S)$,
and if $v\in S$ is open in direction $d$ in $S$,
then $C$ contains the $d$-edge of $\tau(v)$.

The proof is by induction on $|S|=k$.
For $k=1$, $\tau(S)$ is one quadruple and we can take 
its four edges.

Suppose that the statement has been proved 
for $|S|<k$, and assume that $|S|=k$.
Consider an always-connected ordering 
$v_1$, $v_2$, \ldots, $v_k$
of $S$.
Let $S'=S\setminus v_k$.
By the induction hypothesis, 
there is a cycle $C'$ satisfying the requirements.
Vertex $v_k$ has at least one neighbor in $S'$, say, $v_k$ is the neighbor
of $v_i$ in direction $d$. 
But then, $v_i$ is open in direction $d$ in $S'$, therefore,
$C'$ contains the $d$-edge of $\tau(v_i)$.
Remove this edge from $C'$ and substitute by a path of length 5, passing through 
all four vertices of  $\tau(v_k)$. The resulting cycle, $C$, contains all 
vertices of $\tau(S)$, it contains each edge of $\tau(v_k)$, except the one in the opposite direction
to $d$, and it contains all edges of $C'$, except the $d$-edge of $\tau(v_k)$, but 
$v_k$ is not open in $S$ in direction $d$. 
This concludes the induction step, and the proof.
\end{proof}

\begin{proposition}\label{prop:strategyumtocf}
For every $m>1$, $\vp(G_m) \geq \vcs(G_{\floor{m/2}})$.
\end{proposition}
\begin{proof}
Assume, without loss of generality that $m$ is even (if not work
with graph $G_{m-1}$ instead).
In order, to prove that $\vp(G_m) \geq \vcs(G_{\floor{m/2}})$
it is enough, 
given a strategy for player {1} in the
connected set game for $G_{m/2}$,
to construct a strategy for player 1 (the maximizer) in
the path game for $G_m$, so that the result of the path game is
at least as much as the result of the connected set game.
We present the argument as if player 1, apart from the path game, 
plays in parallel a
connected set game on $G_{m/2}$ 
(for which player 1 has a given strategy to choose connected sets in
every round),
where player 1 also chooses the
moves of player 2 in the connected set game.

At round $i$ of the path game on $G_{m}$, player 1 simulates 
round $i$ of the connected set game on $G_{m/2}$. At the start of
round $i$, player 1 has a graph $G^{i-1} \issubgraphof G_{m}$ in the
path game and a graph ${\hat{G}}^{i-1} \issubgraphof G_{m/2}$ in the
connected set game. Player 1 chooses a set ${\hat{S}}^{i}$ in the
simulated connected set game from his given strategy, and then
constructs the path-spanned set $S^i = \tau({\hat{S}}^{i})$ 
(by lemma \ref{lemma:pathspanned}) and
plays it in the path game. Then player 2 chooses a vertex $v_i \in
S^i$. Player 1 computes $\hat{v}_i=\tau'(v_i)$ and simulates the
move $\hat{v}_i$ of player {2} in the connected set game. This is
a legal move for player {2} in the connected set game because
$\hat{v}_i \in {\hat{S}}^{i}$.

We just have to prove that $S^i = \tau({\hat{S}}^{i})$ is a legal
move for player {1} in the path game, i.e.,
$S^i \issubsetof V(G^{i-1})$. We also have to prove 
$S^i=\tau({\hat{S}}^{i})$ 
is spanned by a path in $G^{i-1}$ but this is always true
by lemma \ref{lemma:pathspanned}, since 
${\hat{S}}^{i}$ is a connected vertex set in ${\hat{G}}^{i-1}$. 
Since $S^i \issubsetof \tau(V({\hat{G}}^{i-1}))$,
it is enough to prove that at round $i$, 
$\tau(V({\hat{G}}^{i-1})) \issubsetof V(G^{i-1})$.
The proof is by induction on $i$. 
For $i=1$, $G^0 = G_{m}$, ${\hat{G}}^0 = G_{m/2}$, and thus
$\tau(V({\hat{G}}^{0})) = V(G^{0})$.
At the start of round $i$ with $i>1$, 
$\tau(V({\hat{G}}^{i-1})) \issubsetof V(G^{i-1})$, by the inductive hypothesis.
Then, $\tau(\hat{S}^i) = S^i$
and $\tau(\hat{S}^i \setminus \{\hat{v}_i\}) =
\tau(\hat{S}^i) \setminus \tau(\hat{v}_i) = S^i \setminus \tau(\hat{v}_i)
\issubsetof S^i \setminus \{v_i\}$, because $v_i \in \tau(\hat{v}_i)$. 
Thus, 
$\tau(V({\hat{G}}^{i-1}[\hat{S}^i \setminus \{\hat{v}_i\}])) \issubsetof 
V(G^{i-1}[S^i \setminus \{v_i\}])$, i.e.,
$\tau(V({\hat{G}}^{i})) \issubsetof V(G^{i})$.
\end{proof}

\begin{theorem}\label{theorem:halfgrid}
For every $m > 1$, $\chicf(G_m) \geq \chium(G_{\floor{m/2}})$. 
\end{theorem}
\begin{proof}
By proposition~\ref{prop:vpcf}, $\chicf(G_m) \geq \vp(G_m)$,
by proposition~\ref{prop:strategyumtocf}, $\vp(G_m) \geq
\vcs(G_{\floor{m/2}})$, and by proposition~\ref{prop:cscfplayer1},
$\vcs(G_{\floor{m/2}}) \geq \chium(G_{\floor{m/2}})$.
\end{proof}

\section{Lower bounds on the chromatic numbers of the square grid}%
\label{sec:gridlbum}

Recall that $G_m$ is the $m\times m$ grid graph, that is, the cartesian
product of two paths, each of length $m-1$. 
It was shown in \cite{BCLMZ2009journal} that $\chium(G_m)\ge 3m/2$.
The best known upper bound is $\chium(G_m)\le 2.519m$, from 
\cite{BCLMZ2009sirocconopages,BCLMZ2009journal}.
The main result of this section is the following improvement of the
lower bound.

\begin{theorem}\label{theorem:5mover3lowerbound}
For $m \geq 2$, $\chium(G_m) \geq \tfrac{5}{3}m - \logfhol{m}$.
\end{theorem}
\begin{proof}
For any subset $A\subset V(G)$, let $N_G(A)$ denote the {\em boundary} of $A$,
that is, all vertices which are not in $A$, but neighbors of some vertex in
$A$.
Observe that in a  unique-maximum coloring of a connected graph $G$,
the set of vertices of unique colors form a separator (see, e.g., 
\cite{orderedcoloring}).
Indeed, remove all vertices of unique colors from $G$, let $G'$ be the
remaining graph and let 
color $c$ be the highest remaining color. It is not unique, let $u$ and $v$ be
two vertices of color $c$. Then there can not be a path in $G'$ from $u$ to
$v$, therefore, $G'$ is not connected.

We will use induction on $m$. Consider a  unique-maximum coloring of $G_m$
and take a minimal separator, formed
by vertices of unique colors. 
Using the separator and the coloring, 
after applying a carefully selected sequence of minor operations
(vertex deletion, edge deletion, edge contraction)
on $G_m$,
we obtain
an induced unique-maximum coloring 
(see definition~\ref{defn:minorinducedcoloring}) 
of $G_{m'}$ for some $m'<m$, 
and we apply the induction hypothesis to prove the lower bound. 

%#5

Throughout the proof, we consider $G=G_m$ in its {\em standard drawing}, that
is,
the vertices are points $(x, y)$, $0\le x, y\le m-1$, 
two vertices $(x, y)$ and $(x', y')$ are 
connected if and only if $|x-x'|+|y-y'|=1$, and edges are drawn as straight
line segments. If it is clear from the context, we do not make any notational
distinction
between vertices (edges) and points (resp. segments) representing them.
Denote by $V$ the vertices of the grid, that is, $V=V(G)$. 
Take an additional vertex $v$, ``outside'' $G_m$, say, at $(-2, -2)$,
and connect it with all boundary vertices of $G_m$, so that we do not create
any edge crossing.
Let $G'=G'_m$ denote the resulting graph, Let $V'=V(G')$.

Define graph $H'$ and its drawing as follows. 
The vertex set of $H'$ is $V'$. Vertex $v$ is connected 
to the boundary vertices of the grid, just like in $G'$.
Two vertices, $(x, y)$ and $(x', y')$ in the grid are connected by a straight
line segment in $H'$ 
if and only if 
$|x-x'|\le 1$ and $|y-y'|\le 1$. 

Suppose that $S\subset V'$, and $H'[S]$ contains a non-self-intersecting 
cycle $C$. 
Let $A$ (resp. $B$) 
be those vertices in $V'$ which are {\em inside} (resp. {\em outside})
$C$. If $A, B\ne \emptyset$, then $C$ is called a {\em separating} cycle.
If $A=\emptyset$, then $C$ is called an {\em empty} cycle.
Suppose that $C$ is a separating cycle.
Since edges of $H'$ and edges of $G'$ do not intersect each other, 
$S$ separates $A$ and $B$ in $G'$. 

Suppose now that $S$ is a separator in $G'$
and let 
$A$ be the vertex set of one of the connected components,
separated by $S$.
Clearly, the boundary of $A$,
$N_{G'}(A)$ belongs to $S$, and an 
easy case analysis shows that
the edges of 
$H'[N_{G'}(A)]$, in the present drawing, 
separate the vertices of $A$ from the other vertices.
Suppose from now that $S$ is a minimal separator.
Then, by the previous observations, $H'[S]$ contains one or more separating
cycles.
Let $C$ be a separating cycle in $H'[S]$ with the smallest number of points inside,
and let $A$ be the set of these points.
Then $N_{G'}(A)\subset C$, but since  $N_{G'}(A)$ already separates $A$
from the other vertices, $N_{G'}(A)=S$.
Observe that the only 
empty cycle in $H'$ is the right angled triangle with leg 1.
If  $H'[S]$ contains such a cycle, then one of its vertices 
can be removed from $S$ and we still have a separator. 
Therefore, there are no empty cycles in $H'[S]$.
Moreover, by the minimality of $S$, every separating cycle
in $H'[S]$ contains exactly the points of $A$ in its interior.
It follows, that $H'[S]$ is a cycle that has $A$ in its interior, 
and the remaining points, $V'\setminus (S\cup A)$ in the exterior.

It is easy to see that if $S$ is a separator in $G_m$, then 
$S\cup\{v\}$ is a separator in $G'_m$.
On the other hand, if $S$ is a separator is $G'_m$, then 
$S\setminus\{v\}$ is a separator in $G'_m$. 
Consequently, if $S$ is a minimal separator in $G_m$,
then either $S$ is a minimal separator in $G'_m$, 
or $S\cup\{v\}$ is a minimal separator in $G'_m$.
In the first case we say that $S$ is a {\em cycle-separator}
(see figure~\ref{fig:cyclesep}), 
in the second case we say that it is a {\em path-separator}
(see figure~\ref{fig:pathsep})
of $G_m$.
The vertices of a 
cycle-separator form a cycle in $H'$, and the vertices of a
path-separator form a path, whose first and last vertices 
are the only 
neighbors of $v$, that is, they are on the boundary of the grid,
and the other vertices of $S$ are not on the boundary.

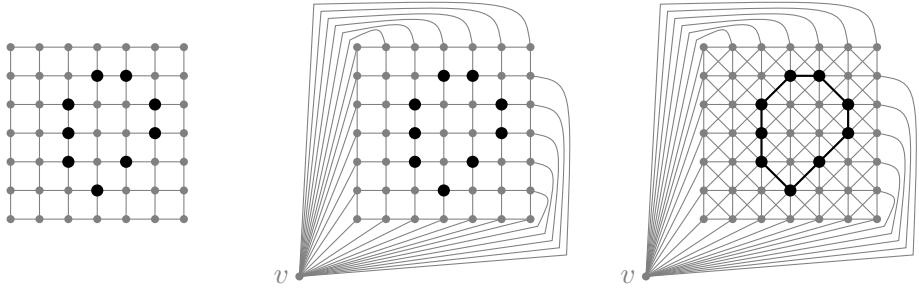
\begin{figure}[tbp]
\centering
\begin{tikzpicture}[scale=0.38]
\newcommand{\normalvsize}{4pt}
\newcommand{\sepvsize}{6pt}
\newcommand{\mx}{6}%
\newcommand{\curvedlines}{%
\foreach \x in {0, ..., \mx} {
    \draw (\x,0) -- (specialv);
}
\foreach \y in {1, ..., \mx} {
    \draw (0,\y) -- (specialv);
}
\foreach \x in {1, ..., \mx} {
    \draw (\x,\mx) .. controls (\x-0.2, 6.58+0.18*\x) .. (-0.25*\x, \mx+0.25*\x)
	 -- 	
	(specialv);
}
\foreach \y in {1, ..., 5} {
    \draw (\mx,\y) .. controls (6.58+0.18*\y, \y-0.2) .. (\mx+0.25*\y,-0.25*\y)
	 -- 	
	(specialv);
}%
}
\newcommand{\diagonallines}{%
    \foreach \x in {1, ..., \mx} {
        \draw (\x,0) -- (0,\x); 
    }%
    \foreach \x in {1, ..., \mx} {
    	\draw (\x,\mx) -- (0,\mx-\x);
    }
    \foreach \x in {1, ..., \mx} {
        \draw (\x,\mx) -- (\mx,\x); 
    }%
    \foreach \x in {1, ..., \mx} {
     	\draw (\x,0) -- (\mx,\mx-\x);
    }
}%
\begin{scope}[xshift=0cm] % 
\begin{scope}[gray]
%\coordinate (specialv) at (-2,-2);
%\fill (specialv) circle (\normalvsize);
%\curvedlines
\foreach \x in {0, ..., \mx} {
    \foreach \y in {0, ..., \mx} {
            \fill (\x,\y) circle (\normalvsize);
    }
    % line drawing
    \draw (\x,0) -- (\x,\mx); % column x
    \draw (0,\x) -- (\mx,\x); % row x
}
\end{scope}
% separator of type i
\fill (3,1) circle (\sepvsize);
\fill (4,2) circle (\sepvsize);
\fill (5,3) circle (\sepvsize);
\fill (5,4) circle (\sepvsize);
\fill (4,5) circle (\sepvsize);
\fill (3,5) circle (\sepvsize);
\fill (2,4) circle (\sepvsize);
\fill (2,3) circle (\sepvsize);
\fill (2,2) circle (\sepvsize);
\end{scope}
\begin{scope}[xshift=12cm] % 
\begin{scope}[gray]
\coordinate (specialv) at (-2,-2);
\fill (specialv) circle (\normalvsize);
\node [left] at (specialv) {$v$};
\curvedlines
\foreach \x in {0, ..., \mx} {
    \foreach \y in {0, ..., \mx} {
            \fill (\x,\y) circle (\normalvsize);
    }
    % line drawing
    \draw (\x,0) -- (\x,\mx); % column x
    \draw (0,\x) -- (\mx,\x); % row x
}
% diagonal lines
%\diagonallines
\end{scope}
% separator of type i
\fill (3,1) circle (\sepvsize);
\fill (4,2) circle (\sepvsize);
\fill (5,3) circle (\sepvsize);
\fill (5,4) circle (\sepvsize);
\fill (4,5) circle (\sepvsize);
\fill (3,5) circle (\sepvsize);
\fill (2,4) circle (\sepvsize);
\fill (2,3) circle (\sepvsize);
\fill (2,2) circle (\sepvsize);
\end{scope}
\begin{scope}[xshift=24cm] % separator of type i (cycle)
\begin{scope}[gray]
\coordinate (specialv) at (-2,-2);
\fill (specialv) circle (\normalvsize);
\node [left] at (specialv) {$v$};
\curvedlines
\foreach \x in {0, ..., \mx} {
    \foreach \y in {0, ..., \mx} {
            \fill (\x,\y) circle (\normalvsize);
    }
    % line drawing
    \draw (\x,0) -- (\x,\mx); % column x
    \draw (0,\x) -- (\mx,\x); % row x
}
% diagonal lines
\diagonallines
\end{scope}
% separator of type i
\fill (3,1) circle (\sepvsize);
\fill (4,2) circle (\sepvsize);
\fill (5,3) circle (\sepvsize);
\fill (5,4) circle (\sepvsize);
\fill (4,5) circle (\sepvsize);
\fill (3,5) circle (\sepvsize);
\fill (2,4) circle (\sepvsize);
\fill (2,3) circle (\sepvsize);
\fill (2,2) circle (\sepvsize);
% simple closed curve induced by separator of type i
\begin{scope}[thick]
\draw (3,1) -- (4,2) -- (5,3) -- (5,4) -- (4,5) --
      (3,5) -- (2,4) -- (2,3) -- (2,2) -- cycle ;
\end{scope}
\end{scope}
\end{tikzpicture}
\caption{A cycle-separator in $G$, $G'$, and $H'$, for $m=7$}
\label{fig:cyclesep}
\end{figure}

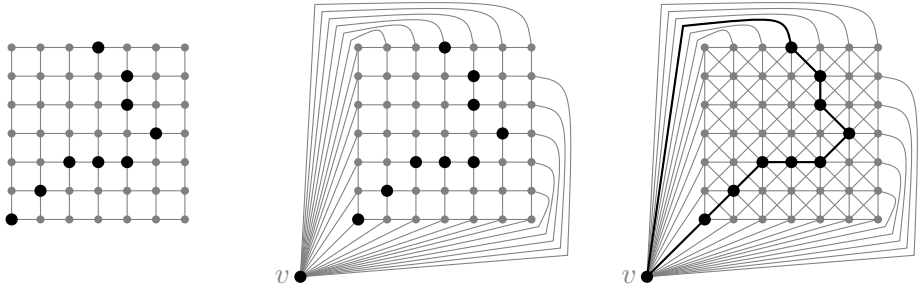
\begin{figure}[tbp]
\centering
\begin{tikzpicture}[scale=0.38]
\newcommand{\normalvsize}{4pt}
\newcommand{\sepvsize}{6pt}
\newcommand{\mx}{6}%
\newcommand{\curvedlines}{%
\foreach \x in {0, ..., \mx} {
    \draw (\x,0) -- (specialv);
}
\foreach \y in {1, ..., \mx} {
    \draw (0,\y) -- (specialv);
}
\foreach \x in {1, ..., \mx} {
    \draw (\x,\mx) .. controls (\x-0.2, 6.58+0.18*\x) .. (-0.25*\x, \mx+0.25*\x)
	 -- 	
	(specialv);
}
\foreach \y in {1, ..., 5} {
    \draw (\mx,\y) .. controls (6.58+0.18*\y, \y-0.2) .. (\mx+0.25*\y,-0.25*\y)
	 -- 	
	(specialv);
}%
}
\newcommand{\diagonallines}{%
    \foreach \x in {1, ..., \mx} {
        \draw (\x,0) -- (0,\x); 
    }%
    \foreach \x in {1, ..., \mx} {
    	\draw (\x,\mx) -- (0,\mx-\x);
    }
    \foreach \x in {1, ..., \mx} {
        \draw (\x,\mx) -- (\mx,\x); 
    }%
    \foreach \x in {1, ..., \mx} {
     	\draw (\x,0) -- (\mx,\mx-\x);
    }
}%
\begin{scope}[xshift=0cm] % 
\begin{scope}[gray]
%\coordinate (specialv) at (-2,-2);
%\fill (specialv) circle (\normalvsize);
%\curvedlines
\foreach \x in {0, ..., \mx} {
    \foreach \y in {0, ..., \mx} {
            \fill (\x,\y) circle (\normalvsize);
    }
    % line drawing
    \draw (\x,0) -- (\x,\mx); % column x
    \draw (0,\x) -- (\mx,\x); % row x
}
\end{scope}
% separator of type ii
\fill (0,0) circle (\sepvsize);
\fill (1,1) circle (\sepvsize);
\fill (2,2) circle (\sepvsize);
\fill (3,2) circle (\sepvsize);
\fill (4,2) circle (\sepvsize);
\fill (5,3) circle (\sepvsize);
\fill (4,4) circle (\sepvsize);
\fill (4,5) circle (\sepvsize);
\fill (3,6) circle (\sepvsize);
\end{scope}
\begin{scope}[xshift=12cm] % 
\begin{scope}[gray]
\coordinate (specialv) at (-2,-2);
\fill (specialv) circle (\normalvsize);
\node [left] at (specialv) {$v$};
\curvedlines
\foreach \x in {0, ..., \mx} {
    \foreach \y in {0, ..., \mx} {
            \fill (\x,\y) circle (\normalvsize);
    }
    % line drawing
    \draw (\x,0) -- (\x,\mx); % column x
    \draw (0,\x) -- (\mx,\x); % row x
}
% diagonal lines
%\diagonallines
\end{scope}
% separator of type ii
\fill (0,0) circle (\sepvsize);
\fill (1,1) circle (\sepvsize);
\fill (2,2) circle (\sepvsize);
\fill (3,2) circle (\sepvsize);
\fill (4,2) circle (\sepvsize);
\fill (5,3) circle (\sepvsize);
\fill (4,4) circle (\sepvsize);
\fill (4,5) circle (\sepvsize);
\fill (3,6) circle (\sepvsize);
% plus vertex v:
\fill (-2,-2) circle (\sepvsize);
\end{scope}
\begin{scope}[xshift=24cm] % separator of type ii (path)
\begin{scope}[gray]
\coordinate (specialv) at (-2,-2);
\fill (specialv) circle (\normalvsize);
\node [left] at (specialv) {$v$};
\curvedlines
\foreach \x in {0, ..., \mx} {
    \foreach \y in {0, ..., \mx} {
            \fill (\x,\y) circle (\normalvsize);
    }
    % line drawing
    \draw (\x,0) -- (\x,\mx); % column x
    \draw (0,\x) -- (\mx,\x); % row x
}
% diagonal lines
\diagonallines
\end{scope}
\fill (specialv) circle (\sepvsize);
% separator of type ii
\fill (0,0) circle (\sepvsize);
\fill (1,1) circle (\sepvsize);
\fill (2,2) circle (\sepvsize);
\fill (3,2) circle (\sepvsize);
\fill (4,2) circle (\sepvsize);
\fill (5,3) circle (\sepvsize);
\fill (4,4) circle (\sepvsize);
\fill (4,5) circle (\sepvsize);
\fill (3,6) circle (\sepvsize);
\begin{scope}[thick]
% simple closed curve induced by separator of type ii
\draw (0,0) -- (1,1) -- (2,2) -- (3,2) -- (4,2) --
      (5,3) -- (4,4) -- (4,5) -- (3,6)  ;
% edges with special vertex
\draw (specialv) -- (0,0);
\draw (3,\mx) .. controls (3-0.2, 6.58+0.18*3) .. (-0.25*3, \mx+0.25*3)
	 -- 	
	(specialv);
\end{scope}
\end{scope}
\end{tikzpicture}
\caption{A path-separator in $G$, $G'$, and $H'$, for $m=7$}
\label{fig:pathsep}
\end{figure}

Our bound is negative for $m\le 64$, so assume that $m>64$, 
and the statement has been proved for 
smaller values of $m$.
Consider an optimal coloring of $G_m$, and let 
$S$ be a minimal separator, all of whose vertices have unique colors.

\smallskip

\noindent {\bf Case 1:} {\em $S$ is a cycle-separator.}
\noindent Let $z$ be the smallest value of $x+y$ over all vertices 
of $S$,  and let $(x, y)$ be the vertex of $S$ for which  $x+y=z$, and
$y$ is the largest. Then vertex $(x+1, y-1)$ is also a vertex of $S$,
and one of $(x,y+1)$, $(x+1,y+1)$ is also in $S$.
Let $(x', y')$ be the vertex of $S$ for which  $x+y=z$, and
$y$ is the smallest. Then $y'< y$, since $(x+1, y-1)$ is in $S$.
Moreover, vertex $(x'-1, y'+1)$ is also a vertex of $S$,
and one of $(x'+1,y')$, $(x'+1,y'+1)$ is also in $S$.
Consider the following contractions of horizontal edges: 
$(x, m-1)(x+1, m-1)$, $(x, m-2)(x+1, m-2)$, $\ldots$, $(x, y)(x+1, y)$, 
$(x+1, y-1)(x+2, y-1)$, $(x+2, y-2)(x+3, y-2)$, $\ldots$, 
$(x', y')(x'+1, y')$, $(x', y'-1)(x'+1, y'-1)$,  $\ldots$, 
$(x', 0)(x'+1, 0)$, and vertical edges:
$(0, y)(0, y+1)$, $(1, y)(1, y+1)$,  $\ldots$, 
$(x, y)(x, y+1)$, $(x+1, y)(x+1, y+1)$,  $(x+2, y-1)(x+2, y)$,
$\ldots$, $(x'+1, y')(x'+1, y'+1)$, $(x'+2, y')(x'+2, y'+1)$, 
$\ldots$,  $(m-1, y')(m-1, y'+1)$.
We obtain a graph, which contains $G_{m-1}$ as a subgraph and the 
induced coloring uses at least two less colors that the coloring of $G_m$
See figure~\ref{fig:branchsets}, where for each gray area, 
vertices are contracted to a single vertex.
The induced coloring uses at least $\chium(G_{m-1})$ colors, therefore, we have 
$\chium(G_m) \geq \chium(G_{m-1}) + 2 \geq 
\frac{5}{3}(m-1) - \logfhol{(m-1)} + 2 > \tfrac{5}{3}m - \logfhol{m}$.

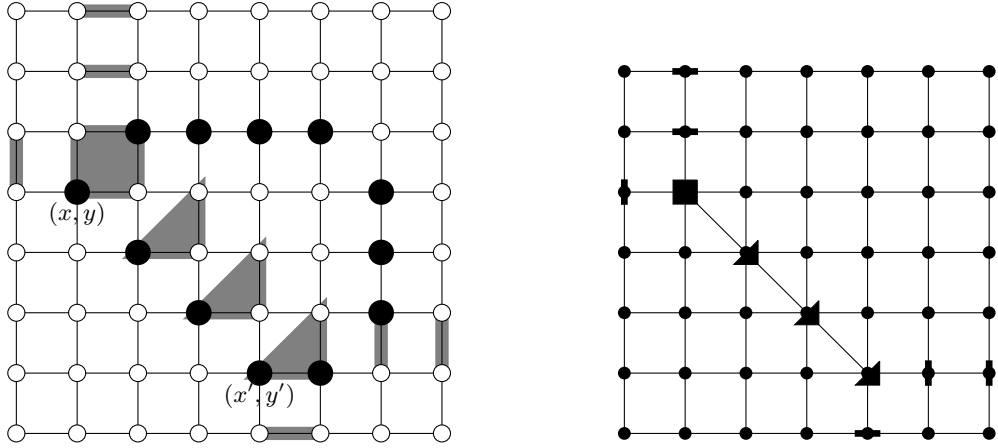
\begin{figure}[hbtp]
\centering
\begin{tikzpicture}[scale=0.8]
\newcommand{\normalvsize}{4pt}
\newcommand{\sepvsize}{6pt}
\newcommand{\hugesize}{8pt}
\newcommand{\fatness}{5pt}
\begin{scope} % separator of type i
\newcommand{\mx}{7}%
\begin{scope}%[gray]
% branch sets 
\begin{scope}[line width=\fatness, gray]
\draw[fill] (4,0) -- (5,0) -- cycle ;
\draw[fill] (6,1) -- (6,2) -- cycle ;
\draw[fill] (7,1) -- (7,2) -- cycle ;
\draw[fill] (0,4) -- (0,5) -- cycle ;
\draw[fill] (1,6) -- (2,6) -- cycle ;
\draw[fill] (1,7) -- (2,7) -- cycle ;
\draw[fill] (4,1) -- (5,1) -- (5,2) -- cycle ;
\draw[fill] (3,2) -- (4,2) -- (4,3) -- cycle ;
\draw[fill] (2,3) -- (3,3) -- (3,4) -- cycle ;
\draw[fill] (1,4) -- (2,4) -- (2,5) -- (1,5) -- cycle ;
\end{scope}
\foreach \x in {0, ..., \mx} {
    % line drawing
    \draw (\x,0) -- (\x,\mx); % column x
    \draw (0,\x) -- (\mx,\x); % row x
}
\foreach \x in {0, ..., \mx} {
    \foreach \y in {0, ..., \mx} {
            \fill[white] (\x,\y) circle (\normalvsize);
    }
    \foreach \y in {0, ..., \mx} {
            \draw (\x,\y) circle (\normalvsize);
    }
}
\end{scope}
% separator of type i
\fill (4,1) circle (\sepvsize);
\fill (5,1) circle (\sepvsize);
\fill (6,2) circle (\sepvsize);
\fill (6,3) circle (\sepvsize);
\fill (6,4) circle (\sepvsize);
\fill (5,5) circle (\sepvsize);
\fill (4,5) circle (\sepvsize);
\fill (3,5) circle (\sepvsize);
\fill (2,5) circle (\sepvsize);
\fill (1,4) circle (\sepvsize);
\fill (2,3) circle (\sepvsize);
\fill (3,2) circle (\sepvsize);
% simple closed curve induced by separator of type i
%\draw (3,1) -- (4,2) -- (5,3) -- (5,4) -- (4,5) --
%      (3,5) -- (2,4) -- (2,3) -- (2,2) -- cycle ;
%
% 
% labels at x, x', y, y'
%\node[above] at (1,7) {column $x$};
\node [below] at (1,4) {\footnotesize $(x,y)$};
\node [below] at (4,1) {\footnotesize $(x',y')$};
\end{scope}
\begin{scope}[xshift=10cm] 
\newcommand{\mx}{6}%
\newcommand{\smallvsize}{3pt}%
\begin{scope}%[gray]
\foreach \x in {0, ..., \mx} {
    \foreach \y in {0, ..., \mx} {
            \fill (\x,\y) circle (\smallvsize);
    }
    % line drawing
    \draw (\x,0) -- (\x,\mx); % column x
    \draw (0,\x) -- (\mx,\x); % row x
}
\end{scope}
% separator of type ii
% simple closed curve induced by separator of type ii
\draw (4,1) -- (3,2) -- (2,3) -- (1,4) ;
% contracted vertices
\begin{scope}%[line width=2pt]
\draw[fill] (1,4)  +(45:\hugesize) -- +(135:\hugesize) --
  +(225:\hugesize) -- +(315:\hugesize) -- +(45:\hugesize) ;
\draw[fill] (2,3)  +(45:\hugesize) -- 
  +(225:\hugesize) -- +(315:\hugesize) -- +(45:\hugesize) ;
\draw[fill] (3,2)  +(45:\hugesize) -- 
  +(225:\hugesize) -- +(315:\hugesize) -- +(45:\hugesize) ;
\draw[fill] (4,1)  +(45:\hugesize) -- 
  +(225:\hugesize) -- +(315:\hugesize) -- +(45:\hugesize) ;
\begin{scope}[line width=2.5pt]
\newcommand{\halfwidth}{6pt}%
\draw (1,5) +(-\halfwidth,0) -- +(+\halfwidth,0);
\draw (1,6) +(-\halfwidth,0) -- +(+\halfwidth,0);
\draw (4,0) +(-\halfwidth,0) -- +(+\halfwidth,0);
\draw (0,4) +(0,-\halfwidth) -- +(0,+\halfwidth);
\draw (5,1) +(0,-\halfwidth) -- +(0,+\halfwidth);
\draw (6,1) +(0,-\halfwidth) -- +(0,+\halfwidth);
\end{scope}
\end{scope}
\end{scope}
\end{tikzpicture}
\caption{Graph $G_m$ with edge contractions and its minor 
    containing $G_{m-1}$}
\label{fig:branchsets}
\end{figure}

%#8

\smallskip

\noindent {\bf Case 2:} {\em $S$ is a path-separator.}
By symmetry we can assume that the path starts in column $x=0$.
If it ends in $x=0$, $y=0$, or in $y=m-1$, then, 
we can remove column $x=0$, and either row $y=0$ or $y=m-1$,
and get a unique maximum coloring of $G_{m-1}$ with at least two less colors.
Then we apply induction as in case 1. 
So we can assume that $S$ ends in $x=m-1$.
It follows that $|S|\ge m$. We distinguish two subcases.

\smallskip

\noindent {\bf Subcase 2.1.} {\em $S$ starts in $x=0$, ends in $x=m-1$, and
$|S|>m$.}

\noindent 
Orient the path formed by the vertices of $S$. For simplicity, call the 
oriented path $v_1, \ldots v_{|S|}$ 
also $S$. The edges of $S$ can be of eight types,
left, right, upper, lower, 
left-upper, left-lower, right-upper, right-lower. 
 
Suppose first that $S$ contains two edges, one of them 
is vertical (left or right edge), 
one of them is horizontal (upper or lower edge),
say, $(x, y)(x+1, y)$ and $(x', y')(x', y'+1)$.
Then contract all edges 
$(x, i)(x+1, i)$, and all edges
$(i, y')(i, y'+1)$, $0\le i\le m-1$, 
to obtain $G_{m-1}$, whose
induced coloring uses at most $\chium(G_{m})-2$ colors. Therefore, we have 
$\chium(G_m) \geq \chium(G_{m-1}) + 2 \geq 
\frac{5}{3}(m-1) - \logfhol{(m-1)} + 2 > 5m/3 - \logfhol{m}$.
So, we can assume in the sequel that either there are no vertical edges, or no
horizontal edges in $S$. Suppose that there are no {\em horizontal} edges, 
and let $v_i=(x, y)$ be a vertex of $S$ where $y$ is the largest. 
Then $v_{i-1}v_i$ is an upper-right edge, and $v_iv_{i+1}$ is a lower-right
edge, or  
$v_{i-1}v_i$ is an upper-left edge, and $v_iv_{i+1}$ is a lower-left
edge. 
We can assume the first one, otherwise we can take the opposite
orientation of $S$. Let $v_i, \ldots , v_j$ be a maximal interval of $S$ 
where all edges are lower-right. 
By assumption,
edge $v_jv_{j+1}$ can not be horizontal, by the minimality of
$S$ 
it can not be upper, 
if it is 
lower, or lower-left, then we can proceed just like in the case 
of cycle-separators, by a sequence of edge contractions we can obtain 
an induced coloring of $G_m$ with two less colors and we are done by
induction.
So, $v_jv_{j+1}$ can only be an upper-right edge. 
We can apply the same argument for the next maximal interval
$v_j, \ldots , v_k$ and obtain that $v_kv_{k+1}$ is a lower-right edge.
We can argue similarly ``backwards'' on $S$, if $v_l, \ldots , v_i$ is a
maximal interval of upper-right edges, then $v_{l-1}v_l$ is a
lower-right edge. It follows, that all edges of $S$ are 
either upper-right, or lower-right. But then $S$ can not have more then $m$
vertices, a contradiction.
In the case when there are no {\em vertical} edges, the argument is 
almost exactly the same.

\smallskip

\noindent {\bf Subcase 2.2.} {\em $S$ starts in $x=0$, ends in $x=m-1$, and
$|S|=m$.}

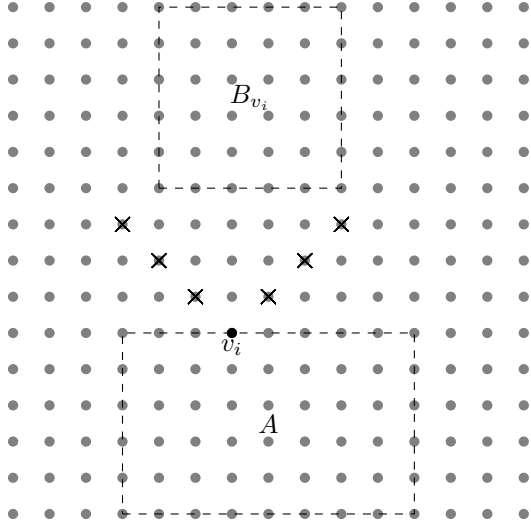
\begin{figure}
\centering
\begin{tikzpicture}[scale=0.48]
\newcommand{\normalvsize}{4pt}
\newcommand{\sepvsize}{6pt}
\newcommand{\mx}{14}%
\pgfsetplotmarksize{8pt}%
\begin{scope} % m = 5k
\begin{scope}[gray]
\foreach \x in {0, ..., \mx} {
    \foreach \y in {0, ..., \mx} {
            \fill (\x,\y) circle (\normalvsize);
    }
}
\end{scope}
% vertex set A
\draw [dashed] (3,0) -- (11,0) -- (11,5) -- (3,5) -- cycle;
% vertex v_i
\fill (6,5) circle (\normalvsize);
% six more vertices in separator
\draw plot[mark=x] (3,8);
\draw plot[mark=x] (4,7);
\draw plot[mark=x] (5,6);
\draw plot[mark=x] (7,6);
\draw plot[mark=x] (8,7);
\draw plot[mark=x] (9,8);
% vertex set B_{v_i}
\draw [dashed] (4,9) -- (9,9) -- (9,14) -- (4,14) -- cycle;
%
% legends
\node at (7.0,2.5) {{\small $A$}};
\node at (6.5,11.5) {{\small $B_{v_i}$}};
\node [below, outer sep = -1pt] at (6,5) {{\small $v_i$}};
%
%\node at (7.0,-1.4) {(a)};
\end{scope}
\end{tikzpicture}
\caption{The subcase $\card{S} = m$}
\label{fig:msep}
\end{figure}

\noindent 
Since If $|S| = m$, $S=v_1, v_2, \dots, v_m$ 
such that $v_i = (i-1,y_i)$, for every $i$.
We show that $G_m\setminus S$ contains a subgraph isomorphic to
$G_{2k}$.

Suppose that $5k\le m\le 5k+4$.
Consider the set of vertices
\[A = \{(x,y) \mid  
   \text{$k \leq x \leq 4k-1$, $0 \leq y \leq 2k-1$}\}.\]
Set $A$ induces a $3k \times 2k$ grid graph, $G_{3k,2k}$, in $G_m$.
If $A \cap S = \emptyset$, then $G_m-S \containssubgraph G_{3k,2k}
\containssubgraph G_{2k}$; otherwise some $v_i \in S$ belongs to
$A$, i.e., $v_i = (i,y_i)$
with $k \leq i \leq 4k-1$ and $0 \leq y_i \leq 2k-1$. 
Then, consider the set of vertices
\[B_{v_i} = \{(x,y) \mid 
   \text{$i-k+1 \leq x \leq i+k$, $3k \leq y \leq m-1$}\},\]
which contains a $G_{2k}$ subgraph in $G_m$ and it is disjoint from $S$.
Therefore, $G_m - S$ contains a subgraph isomorphic to
$G_{2k}$, and thus
$\chium(G_m) \geq 
 m + \chium(G_{2k})\ge m + \tfrac{10}{3}k -  \logfhol{2k}
\ge \tfrac{5}{3}m - \logfhol{m}$.
\end{proof}

\begin{remark}
By a slightly more careful calculation we could get
$\chium(G_m) \geq \tfrac{5}{3}m-\log_{5/2}m$.
\end{remark}

An immediate corollary from theorem~\ref{theorem:halfgrid} 
is the following.

\begin{corollary}
For $m \geq 2$, 
$\chicf(G_m) \geq \tfrac{5}{6}m - 10\log_2{m}$.
\end{corollary}

\section{Discussion and open problems}\label{sec:conclusion}

As we mentioned in the introduction, conflict-free and
unique-maximum colorings can be defined for hypergraphs. In the
literature of conflict-free colorings, hypergraphs that are induced
by geometric shapes have been in the focus. It would be
interesting to show possible relations of the respective chromatic
numbers in this setting.

An interesting open problem is to determine the exact value of the
unique-maximum chromatic number for the square grid $G_m$. In this
paper, we improved the lower bound asymptotically to $5m/3$,
and we believe that this bound is still far from optimal. 
Observe that in each case, our recursion step
would allow us to prove  
a lower bound of the form $2m-o(m)$, with the exception of the last 
case, when $|S|=m$, that is the ``bottleneck'' of the proof.
We believe that using a more complicated recursion, with grids of rectangular 
shapes, could lead to 
an improvement.

Another area for improvement is the relation between the two
chromatic numbers for general graphs. We have only found graphs which
have unique-maximum chromatic number about twice the conflict-free
chromatic number, but the only bound we have proved on $\chium(G)$
is exponential in $\chicf(G)$.

Finally, the coNP-completeness of checking whether a coloring is 
conflict-free, implies that the decision problem for the
conflict-free chromatic number is in complexity class $\Pi_2^p$
(at the second level of the polynomial hierarchy). An interesting
direction for research would be to attempt a proof of $\Pi_2^p$
completeness for this last decision problem.

\bibliographystyle{plain}
\bibliography{cfcolor}

\newpage
\appendix

\end{document}